\documentclass[10pt]{article}
\usepackage{color, amsmath, amssymb, amsbsy, amsthm, graphicx, bbm, amsfonts, bm,dsfont,enumitem,booktabs}
\usepackage{caption}   %
\usepackage{booktabs}  %
\usepackage{setspace, comment, graphicx, latexsym,titlesec}
\usepackage{multirow}
\usepackage{array}
\usepackage[table]{xcolor}

\usepackage[top=1in, bottom=1in, left=1in, right=1in]{geometry}
\usepackage[flushleft]{threeparttable}
\usepackage[colorlinks,allcolors=blue]{hyperref}
\usepackage[square,numbers]{natbib}  
\usepackage{titlesec}
\allowdisplaybreaks
\usepackage[ruled,vlined,linesnumbered]{algorithm2e}
\usepackage{wrapfig,lipsum,booktabs}
\newtheoremstyle{theoremstyle} 
{\topsep} 
{\topsep} 
{\itshape} 
{} 
{} 
{} 
{.5em} 
{\color{black}\ifthenelse{\equal{#3}{}}{{\bfseries #1 #2}}{{\bfseries #1 #2 (#3)}}}
\newtheoremstyle{examplestyle}
{\topsep} 
{\topsep} 
{} 
{} 
{} 
{} 
{.5em} 
{\color{black}\ifthenelse{\equal{#3}{}}{{\bfseries #1 #2}}{{\bfseries #1 #2 (#3)}}}
\theoremstyle{theoremstyle}\newtheorem{theorem}{Theorem}
\theoremstyle{theoremstyle}     
\theoremstyle{theoremstyle}\newtheorem{lemma}{Lemma}  
\theoremstyle{theoremstyle}        
\theoremstyle{theoremstyle}
\theoremstyle{theoremstyle}\newtheorem{assumption}{Assumption}
\theoremstyle{theoremstyle}

\theoremstyle{examplestyle}

\theoremstyle{examplestyle}\newtheorem{remark}{Remark}

\def \hat{\widehat}
\def \tilde{\widetilde}

\begin{document}

\title{\Large HERO: Improving the Reliability and Sensitivity of Generative Model Evaluation Using Historical Data}

\title{\Large HERO: Improving the Reliability and Sensitivity of Generative Model Evaluation Using Historical Data}

\renewcommand{\thefootnote}{}
\author{
Xinrui Ruan\textsuperscript{1,2,*} \and
Zhenyu Zhao\textsuperscript{2,*} \and
Waverly Wei\textsuperscript{3} \and
Yueshan Zhang\textsuperscript{4} \and
Zeyu Zheng\textsuperscript{5} \and
Sui Huang\textsuperscript{2}\and
Jingshen Wang\textsuperscript{1,\textdagger}
}
\renewcommand{\thefootnote}{\arabic{footnote}}

\date{
{\small
\textsuperscript{1}Division of Biostatistics, University of California, Berkeley\\
\textsuperscript{2}Roblox Corporation\\
\textsuperscript{3}Department of Data Sciences and Operations, University of Southern California\\
\textsuperscript{4}School of Mathematical Sciences, Nankai University\\
\textsuperscript{5}Department of Industrial Engineering and Operations Research, University of California, Berkeley
}
}

\maketitle

\let\thefootnote\relax
\footnotetext{
\begin{tabular}{@{}l@{\hspace{0.3em}}l@{}}
\textsuperscript{*} & Xinrui Ruan and Zhenyu Zhao contributed equally. This work was done while Xinrui Ruan was an intern at Roblox Corporation.\\
\textsuperscript{\textdagger} & Correspondence: jingshenwang@berkeley.edu.
\end{tabular}
}
\let\thefootnote\arabic

\begin{abstract}
\noindent Reliable generative AI models critically rely on expert human annotations to evaluate output quality, yet these ``gold'' labels are expensive to collect and limited in quantity. Organizations thus often turn to collecting vast but noisy ``silver'' labels from crowdsourced workers or vendor annotators as proxies for gold labels. Because gold remains the evaluation target, naively aggregating noisy silver labels may introduce bias, and estimators built on sparsely observed gold labels may have high variance to resolve the model performance gaps that guide practical decisions. 
Model evaluation has become an ongoing operational practice rather than a one-time exercise, with evaluation rounds repeating across model versions, releases, and content domains. A natural question is whether the previous historical evaluation data can be used to improve each new round of evaluation.  We introduce HERO (History Enhanced RObust model evaluation), a novel framework that uses historical data to suppress bias (improve reliability) and reduce variance (improve sensitivity) in model performance evaluation. HERO calibrates silver labelers' performance learned from historical gold annotations, and stabilizes the resulting estimator by anchoring it to covariate information measured with high precision in the historical data.  HERO can be broadly applied across multiple common evaluation tasks, and remains valid when only a subset of historical labelers appears in the current round. We establish conditions under which the bias and variance reductions hold, showcase HERO's performance in simulation studies, and demonstrate its effectiveness on real-world model evaluation benchmarking datasets.

\end{abstract}

\thispagestyle{empty}
\clearpage

\thispagestyle{plain}\setcounter{page}{1}

\doublespacing

\section{Introduction}\label{sec:intro}

\subsection{Motivation and Contribution}\label{subsec:motivation}
Human evaluation is at the heart of developing and deploying reliable generative AI models \citep{chiang2024chatbot}. Many organizations rely heavily on human annotations to evaluate model output quality, assess safety, calibrate moderation systems, and compare candidate model versions before release \citep{snow2008cheap,northcutt2021confident,frenay2013classification}. Small differences in estimated model quality, even on the order of a few percentage points in win rate or defect rate, determine which model is shipped to hundreds of millions of users \citep{weidinger2025toward,stein2023exposing}. 

As widely documented in natural language processing and related fields, a key challenge within an annotation task is that expert-annotated ``gold" labels (a.k.a. ground truth labels) are expensive and scarce \citep{snow2008cheap}, so organizations often instead rely on judgments from a pool of less specialized ``silver" annotators whose labeling reliability varies. The so-called \textit{gold} labels may be produced by a single expert annotator or by consensus among multiple experts within the organization, and they serve as the best available reference standard. In contrast, large-scale evaluations typically rely on \emph{silver} labels collected from broader, less specialized sources, including crowd or vendor annotators and automated judges such as LLMs \citep{gu2024survey} or VLMs \citep{chen2024mllm}. Throughout this paper, we view silver labels as the primary operational labels used in evaluation, and gold labels as a more accurate but limited supervision source used to assess and calibrate them. Because silver annotators differ in reliability, simply aggregating their labels without adjustment can reduce the reliability of evaluation and introduce bias in the estimated model performance.

Many methods have been developed to account for annotator variability \citep{DawidSkene1979,snow2008cheap,whitehill2009whose,raykar2010learning}, but most operate within a single evaluation round and do not fully leverage the unique structure of modern model evaluation settings. We note that two novel features distinguish modern AI model evaluations from classical annotation tasks. \textit{First}, organizations that evaluate generative models accumulate large-scale historical annotation data across prior evaluation rounds. These datasets often include repeated observations of annotator behavior and subsets of instances with gold labels, which together reveal patterns in labeler reliability. \textit{Second}, the historical and current model evaluation tasks may share content that transfers across rounds, such as prompt features.  Generation prompts often persist across rounds because organizations evaluate successive model versions on a stable prompt distribution drawn from real user traffic or curated benchmarks \citep{chiang2024chatbot}, so the same prompt features that appeared in earlier rounds also characterize the current evaluation set. With these two features, in principle, labeler reliability learned from historical gold labels can calibrate current annotations, and the shared generation prompts available in a large historical sample can stabilize current model evaluation estimates. We defer a more complete comparison with existing literature to Appendix F.

To fully exploit historical data for reliable model evaluation, we propose HERO (\textbf{H}istory \textbf{E}nhanced \textbf{R}\textbf{O}bust model evaluation), a framework that adjusts noisy annotation-based estimators with historical evaluation data to reduce both bias and variance in model performance estimates. The key contributions of our work include:
\vspace{-0.2cm}
\begin{itemize}[left=0.4cm]
\setlength\itemsep{-0.2em}
\item We elucidate two unique features that distinguish modern AI evaluation from classical annotation literature (Figure 2), namely the accumulation of historical annotation data with gold-labeled subsets, and prompt feature transfer from earlier rounds to the current evaluation round.
\item We identify bias sources that arise in this evaluation setting through real-world benchmarking data, isolating heterogeneous silver labeler reliability and labeler pool shift between rounds (Figure 1).
\item We propose a theoretically grounded yet simple method, HERO, that combines bias reduction from labeler reliability estimation with variance reduction from control variates  (Section \ref{sec:method}). Our control variate construction is customized at the score level, departing from the existing literature that operates on raw outcomes.
\item We show that the proposed method applies across multiple common model evaluation use cases (Table 1) and validate our method with comprehensive empirical evidence (Sections \ref{sec:simulation} and \ref{sec:real world}).
\end{itemize}

\subsection{Model Evaluation Target Parameters and Tasks}\label{sec:target_params}

We organize common model evaluation scenarios into a unified observational structure. Consider $n$ instances indexed by $i = 1, \ldots, n$. For each instance, $Y_i^*$ denotes the \textit{operational ground-truth gold label}, the best available adjudicated assessment of the generated output, and $Z_i$ denotes the noisy \textit{silver labels} collected from less specialized annotators or automated AI systems as proxies for $Y_i^*$. The target parameter is a function of $Y_i^*$, and the task-specific outcome $Y_i$ is constructed from $Y_i^*$ according to the scenario. Table~\ref{tab:use_cases} summarizes four evaluation scenarios, listing the outcome variable $Y_i$, the corresponding target parameter, and the available covariates for each. We summarize each case in Table 1 below and defer its detailed explanation to the appendix. 

\begin{table}[h]
\footnotesize 
\resizebox{\textwidth}{!}{%
\begin{tabular}{>{\centering\arraybackslash}p{3.4cm} 
                |>{\centering\arraybackslash}p{4cm} 
                |>{\centering\arraybackslash}p{4cm}}
\hline
\rowcolor{gray!20}
\textbf{Use Case} & \textbf{Outcome $Y_i$} & \textbf{Target Parameter}  \\
\hline
{1. Single Model Evaluation}
& Gold label for Model A: $Y_i = Y^*_{i,\texttt{A}}$ 
& Average performance: \newline $\mu_{\texttt{A}} = \mathbb{E}[Y_i]$ \\ \hline

{2. Model Comparative Evaluation} 
& Difference in gold labels: $Y_i = Y^*_{i,\texttt{A}} - Y^*_{i,\texttt{B}}$
& Average performance difference:  $\alpha = \mathbb{E}[Y_i]$ \\ \hline

{3. Side-by-Side Model Comparison} 
& Binary preference indicator: $Y_i = \mathds{1}(Y^*_{i,\texttt{A}} \succ Y^*_{i,\texttt{B}})$
& Win rate for Model A:  $\pi_{\texttt{A}\succ\texttt{B}} = \mathbb{E}[Y_i]$ \\ \hline

{4. Safety Moderation Calibration} 
& True positive flag: $Y_i = Y_i^* Y_i^{\texttt{AI}}$ 
& Precision $= {\mathbb{E}[Y_i]}/{\mathbb{E}[Y_i^{\texttt{AI}}]}$; \newline Recall $= {\mathbb{E}[Y_i]}/{\mathbb{E}[Y_i^*]}$ \\ \hline
\end{tabular}%
}
\caption{Summary of model evaluation use cases.\label{tab:use_cases} }\vspace{-0.5cm}
\end{table}

\begin{figure}[h]
    \centering
    \includegraphics[width=0.8\linewidth]{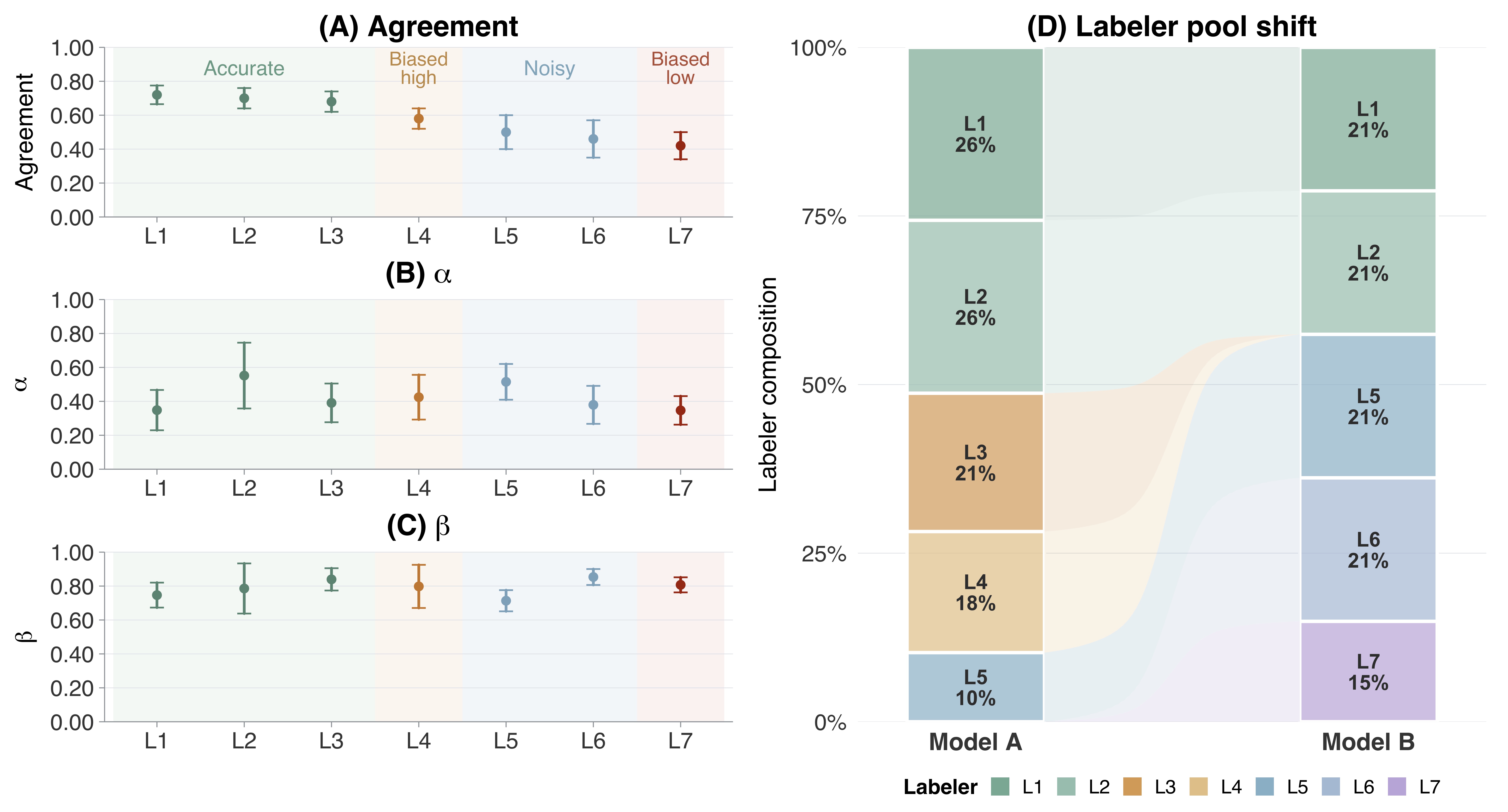}
    \caption{Panels (A)--(C) show heterogeneous silver labeler agreement rates, sensitivity $\alpha_j$, and specificity $\beta_j$, with 95\% confidence intervals in our benchmarking data (Appendix F, Asset detect detection data). Panel (D) shows the labeler pool shift between the two model evaluation rounds.} \vspace*{-0.2cm}
    \label{fig:bias_sources}
\end{figure}

\section{Bias and Variance in Model Evaluation}\label{sec:bias_variance}

Estimating the target parameter $\mathbb{E}[Y_i]$ from gold labels alone is rarely feasible in practice. This is because many evaluation tasks provide no gold labels at all. Even when a subset of instances is annotated by a gold labeler, the labeled fraction is typically small. The sample mean over this subset is unbiased, but its variance remains too large to resolve the differences in model performance that guide practical decisions. For these reasons, practitioners often construct a proxy for $Y_i$ by aggregating the available silver labels $Z_i$. 
Let $\hat{\mu}$ denote an estimator of $\mathbb{E}[Y_i]$ constructed from a silver-labeler-based proxy. To assess estimator quality, we decompose the mean-squared error (MSE) into bias and variance:
\begin{equation}
\underbrace{\mathbb{E}\!\Big[
  \big(\hat{\mu} - \mathbb{E}[Y_i]\big)^2
\Big]}_{\text{MSE relative to target}}
\;=\;
\underbrace{\big(\mathbb{E}[\hat{\mu}]
  - \mathbb{E}[Y_i]\big)^2}_{\text{Bias}^2}
\;+\;
\underbrace{\mathrm{Var}(\hat{\mu})}_{\text{Variance}}.
\label{eq:mse_decomp}
\end{equation}

Both components introduce distinct challenges. We start by examining the bias component through two sources that arise in generative model evaluation. The first source of bias is annotator heterogeneity. The inter-rater reliability literature consistently finds variation in annotator accuracy across domains from medical imaging
\citep{Warfield2004, Raykar2010} to natural language processing
\citep{Snow2008}, even under shared guidelines
\citep{Cohen1960, Fleiss1971, DawidSkene1979, Whitehill2009}. Differences in expertise, interpretation standards, and
decision thresholds cause some labelers to favor precision over
recall, while others show the opposite tendency. Aggregation
without calibration, such as the majority vote-based estimators, allows these systematic tendencies,
particularly from lower-quality labelers, to shift the estimate
of $\mathbb{E}[Y_i]$ away from the true target.

As shown in Figure 1 Panels A-C, the seven silver labelers in our benchmarking data span agreement rates from below 0.4 to above 0.8, and the sensitivity $\alpha_j$ and specificity $\beta_j$ panels reveal asymmetric error profiles behind that spread, with some labelers favoring positives, others favoring negatives, and a third group noisy on both axes. A majority-vote estimator weights these labelers equally and therefore inherits the mixture rather than the gold-labeler target.

A second source of bias arises when the composition of the
labeler pool changes between evaluation rounds. For example, in comparative
evaluation tasks (Use Case~2 in Table~\ref{tab:use_cases}),
organizations may evaluate a newer model against an earlier
version at different times. The labeler pools that annotate each
model may only partially overlap. A reliable annotator who
contributes a large fraction of Model~A labels may complete only
a small fraction of Model~B labels, and vice versa. When the
proportion of reliable annotators differs between the two rounds,
the observed label difference reflects both true differences in
model quality and shifts in the composition of the
labeler pool. An estimator that treats the two label sets as
exchangeable confounds these two sources of variation. Figure~1 Panel D illustrates this shift in our benchmarking data. Two labelers active in the Model~A round do not appear in the Model~B round, and the labelers active in both rounds contribute different proportions of labels across the two models.

The variance component of the MSE in
Equation~\eqref{eq:mse_decomp} remains a challenge even after
addressing the bias issues mentioned above. Even with a well-calibrated silver label based proxy, a small labeled set or high labeling outcome heterogeneity across instances can produce confidence intervals too wide to resolve the
performance gaps that guide practical decisions. A variance reduction technique that halves the estimator variance is equivalent to doubling the effective sample size, so that an evaluation previously requiring 1{,}000 labeled instances can reach the same precision with 500, cutting annotation cost and turnaround time proportionally.

\section{Method: History Enhanced RObust (HERO) Model Evaluation}\label{sec:method}

\begin{figure}[t]
    \centering
    \includegraphics[width=\linewidth]{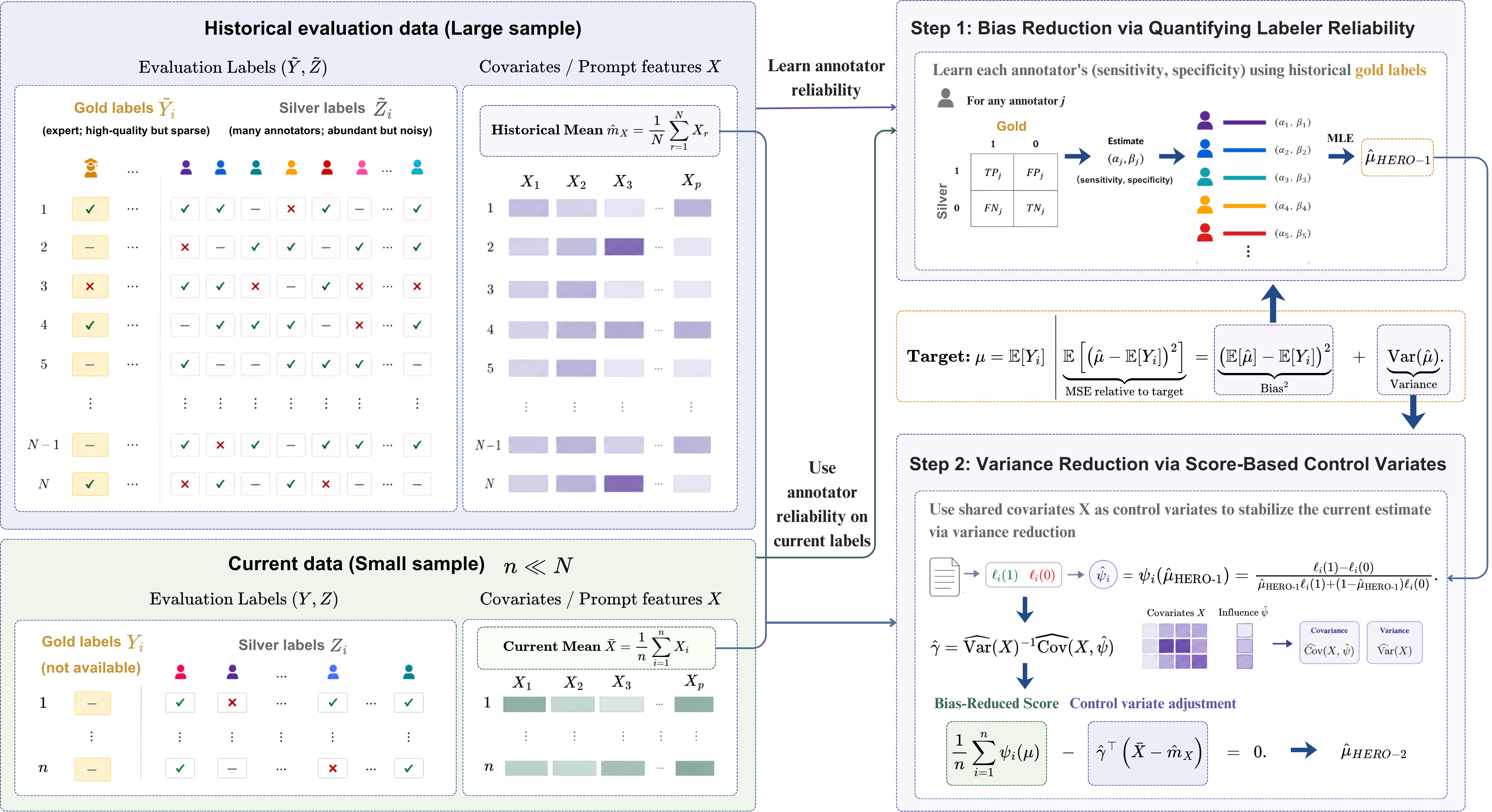}
\caption{
{Overview of the History-Enhanced Robust (HERO) evaluation framework.}
        HERO combines a large historical dataset $(\tilde{Y}_i, \tilde{Z}_i, \tilde{X}_i)_{i=1}^N$ with a small current sample $(Z_i, X_i)_{i=1}^n$ ($n \ll N$) to estimate the target parameter $\mu = \mathbb{E}[Y_i]$ with $Y_i$ unobserved on current evaluation data. 
         \vspace*{-5mm} 
    }
    \label{fig:hero_diagram}
\end{figure}

In this paper, we propose a unified framework to reduce both bias and variance for an \textit{accurate} (improved reliability) and \textit{efficient} (improved sensitivity) estimation of the target parameter $\mu$. Achieving both goals requires an external source of information beyond the only current silver labels. As organizations that evaluate generative models often maintain historical evaluation datasets from prior rounds, our framework leverages these historical data to construct an estimator of $\mu$ with reduced MSE. We term the proposed framework History Enhanced RObust (HERO) model evaluation and organize it as a two-step procedure that first reduces the bias and then controls the variance.

At a very high level, HERO works as follows. When we analyze current evaluation data with only silver labels, HERO first constructs calibrated gold proxies that weight each silver labeler by the reliability learned from historical data, yielding a bias-reduced estimator $\hat{\mu}_{\text{HERO-1}}$. HERO then
incorporates historical covariates (e.g., prompt features) as
control variates to further reduce the variance of
$\hat{\mu}_{\text{HERO-1}}$, yielding the final estimator
$\hat{\mu}_{\text{HERO-2}}$. This estimator satisfies:
\vspace{-0.2cm}
\begin{itemize}[left=0.4cm]
\setlength\itemsep{-0.25em}
    \item $\hat{\mu}_{\text{HERO-1}}$ and $\hat{\mu}_{\text{HERO-2}}$ both have reduced bias compared
    to the majority-vote-based estimator.
    \item The variance of $\hat{\mu}_{\text{HERO-2}}$ is
    reduced compared to $\hat{\mu}_{\text{HERO-1}}$ when shared covariates between the current and historical data are available.
\end{itemize}

\subsection{Historical and Current Evaluation Data}\label{subsec:notation}

To distinguish historical data from the current evaluation task, we introduce a consistent notation. We use a tilde to denote quantities associated with historical data. The historical dataset consists of $N$ instances collected from earlier evaluation rounds and is defined as follows:
\vspace{-0.25cm}
\begin{itemize}[left=0.4cm]
\setlength\itemsep{-0.2em}
 \item Noisy silver labels $\tilde{Z}_i$ from a pool of $m$ silver labelers.  Not all labelers annotate every instance.  The availability indicator $\tilde{A}_{ij} \in \{0,1\}$ records whether labeler~$j$ provided a label for
  instance~$i$.
  
 \item A set of instances together with their associated covariates $\{\tilde{X}_i\}_{i=1}^N$, such as prompt features, metadata, or contextual information. These prompts may \textit{reappear} in future evaluation tasks.

\item Gold labels $\tilde{Y}_i$ are collected on a subset of historical instances to audit silver labeler performance and identify which labelers are consistently accurate. These gold labels are obtained under a randomized routing mechanism with known selection probability given $\tilde{X}_i$.  The routing indicator $\tilde{R}_i \in \{0,1\}$ equals $1$ when instance~$i$ is sent to a gold labeler and $0$ otherwise.  Gold labels are observed only when $\tilde{R}_i = 1$.
\end{itemize}

Meanwhile, the current evaluation dataset contains $n$ instances, where $n \ll N$: 
\vspace{-0.25cm}
\begin{itemize}[left=0.4cm]
\setlength\itemsep{-0.2em}
\item A set of instances together with their associated covariates
  $\{X_i\}_{i=1}^{n}$, including prompt features, metadata, and contextual information.
\item No gold labels are observed in the current task.
\item Noisy silver labels $Z_i$ with availability indicators
  $A_{ij} \in \{0,1\}$.
\end{itemize}
\vspace{-0.3cm}

Throughout this section, all outcomes and labels are binary, so that $Y_i \in \{0,1\}$ and $\tilde{Y}_i \in \{0,1\}$.  The target parameter remains $\mu = \mathbb{E}[Y_i]$ as defined in Section~\ref{sec:target_params}, but gold label $Y_i$ is unobserved.

\noindent \textit{Scope of historical data}. Our framework allows two kinds of mismatch between current and historical data. First, the current annotator pool can differ, but it is a subset of the labelers who contributed to the historical data. Second, the historical task itself may differ from the current task, since HERO calibrates labeler reliability rather than item-level labels.

\subsection{HERO Step 1: Bias Reduction via Quantifying Labeler Reliability}\label{subsec:hero1}

To account for model evaluation bias introduced by silver labeler heterogeneity, historical data provide a natural source for revealing persistent patterns in labeler behavior. HERO Step 1 extracts these reliability patterns from the historical data and transfers them to the current evaluation task, where no gold labels are available. More concretely, we characterize each silver labeler $j$ by two parameters that describe intrinsic properties of that labeler's decision process,
\begin{align*}
\alpha_j &= P(Z_{ij}=1 \mid Y_i=1), \quad \beta_j  = P(Z_{ij}=0 \mid Y_i=0).
\end{align*}
Here $\alpha_j$ is the sensitivity and $\beta_j$ is the specificity. Because these parameters describe the labeler rather than the task, they are expected to transfer from the historical dataset to the current dataset even when the underlying evaluation content/task differs.

HERO estimates $(\alpha_j, \beta_j)$ for each labeler $j$ from the historical dataset
$\{(\tilde{X}_i, \tilde{R}_i\tilde{Y}_i, \tilde{Z}_i, \tilde{R}_i, \tilde{A}_{ij})\}_{i=1}^{N}$ using an EM algorithm that treats unobserved gold labels as latent variables, and use observed gold labelers as anchors. Given a candidate value $y \in \{0,1\}$ for the gold label, the likelihood of
the observed silver labels for historical instance $i$ is 
\begin{equation}
  \tilde{\ell}_i(y)
  = \prod_{j \,:\, \tilde{A}_{ij}=1} P(\tilde{Z}_{ij} \mid \tilde{Y}_i = y),
  \label{eq:hist_instance_lik}
\end{equation}
where $P(\tilde{Z}_{ij} \mid \tilde{Y}_i = 1)
= \alpha_j^{\,\tilde{Z}_{ij}}\,(1 - \alpha_j)^{1 - \tilde{Z}_{ij}}$
and $ P(\tilde{Z}_{ij} \mid \tilde{Y}_i = 0)= (1 - \beta_j)^{\tilde{Z}_{ij}}\,\beta_j^{\,1 - \tilde{Z}_{ij}}$. 
An instance without a gold label
($\tilde{R}_i = 0$) contributes the marginal
$\tilde{\mu}\,\tilde{\ell}_i(1) + (1-\tilde{\mu})\,\tilde{\ell}_i(0)$, where
$\tilde{\mu} = \mathbb{E}[\tilde{Y}]\neq \mathbb{E}[Y]$. The full historical log-likelihood is
\begin{equation}
  \tilde{\mathcal{L}}\bigl(\tilde{\mu},\,\{\alpha_j,\beta_j\}_{j=1}^{m}\bigr)
  = \sum_{i=1}^{N}\Bigl[
      \tilde{R}_i\,\log \tilde{\ell}_i(\tilde{Y}_i)
      + (1 - \tilde{R}_i)\,
        \log\bigl(\tilde{\mu}\,\tilde{\ell}_i(1)
        + (1 - \tilde{\mu})\,\tilde{\ell}_i(0)\bigr)
    \Bigr].
  \label{eq:hist_loglik}
\end{equation}

To stabilize the EM updates, HERO initializes the parameters using only the gold-labeled subset
$\mathcal{G}=\{i:\tilde{R}_i=1\}$. Specifically, we set
\[
\tilde{\mu}^{(0)}=|\mathcal{G}|^{-1}\sum_{i\in\mathcal{G}}\tilde{Y}_i,
\qquad
\alpha_j^{(0)}
=
\frac{\sum_{i\in\mathcal{G}}\tilde{A}_{ij}\tilde{Y}_i\tilde{Z}_{ij}}
     {\sum_{i\in\mathcal{G}}\tilde{A}_{ij}\tilde{Y}_i},
\qquad
\beta_j^{(0)}
=
\frac{\sum_{i\in\mathcal{G}}\tilde{A}_{ij}(1-\tilde{Y}_i)(1-\tilde{Z}_{ij})}
     {\sum_{i\in\mathcal{G}}\tilde{A}_{ij}(1-\tilde{Y}_i)}.
\]
HERO then maximizes
$\tilde{\mathcal{L}}$ via the EM algorithm (see Appendix Algorithm~1,
Stage~A), producing estimates $(\hat{\alpha}_j,\hat{\beta}_j)$ that are
treated as fixed when estimating $\mu$ on the current data.

The current-data likelihood $l_i(y)$ follows the same structure as defined in \eqref{eq:hist_instance_lik}, with $(Z_{ij}, A_{ij})$
replacing $(\tilde{Z}_{ij}, \tilde{A}_{ij})$ and the estimates
$(\hat{\alpha}_j, \hat{\beta}_j)$ replacing the unknown parameters. Because no
gold labels are observed, every instance contributes through the marginal form.
The HERO Step~1 estimator maximizes the log-likelihood over $\mu \in (0,1)$,
\begin{equation}
  \hat{\mu}_{\textsc{hero-1}}
  = \arg\max_{\mu \in (0,1)}\;
    \sum_{i=1}^{n}
    \log\bigl[\mu\,\ell_i(1) + (1 - \mu)\,\ell_i(0)\bigr].
  \label{eq:hero1}
\end{equation}
The first-order condition is
\begin{equation}
  \sum_{i=1}^n \psi_i(\mu):=\sum_{i=1}^{n}
  \frac{\ell_i(1) - \ell_i(0)}
       {\mu\,\ell_i(1) + (1 - \mu)\,\ell_i(0)}
  = 0,
  \label{eq:score_eq}
\end{equation}
which can be solved by any standard one-dimensional optimizer because $\mu$ is
scalar. Here, $\psi_i(\mu)$ denotes the individual score contribution from instance $i$.

\begin{remark}[Comparison with majority vote]
\label{rem:bias_reduction}
The instance-level ratio $\ell_i(1)/\ell_i(0)$ controls how each instance
influences $\hat{\mu}_{\textsc{hero-1}}$. When most high-sensitivity labelers
agree on a positive label, this ratio is large, pulling the estimate toward the
value implied by accurate annotators. Conversely, agreement among labelers with
low specificity on a negative label carries less weight. When all labelers share
identical accuracy, $\ell_i(1)/\ell_i(0)$ reduces to a monotone function of the
majority-vote count, and the estimator recovers a majority-vote-type solution as
a special case.
\end{remark}

\subsection{HERO Step 2: Variance Reduction via Score-Based Control Variates}\label{subsec:hero2}

On top of bias reduction, HERO can further reduce estimator variance by leveraging auxiliary covariates that are shared between the historical and current datasets. Prompt features are a natural shared covariate, since organizations typically evaluate successive model versions on a stable prompt distribution drawn from real user traffic or curated benchmarks, and the same prompt features that appeared in earlier rounds also characterize the current evaluation set. 

The key observation is that the historical dataset ($N \gg n$), being much larger than the current evaluation sample, provides a substantially more precise estimate of the covariate mean. This historical mean can therefore serve as a stable reference in a control-variate adjustment without changing the target parameter. 

Yet, the classical control variate adjustment does not transfer directly in our task, because the Step 1 estimator is a calibrated functional of EM posteriors rather than a sample mean. We thus propose to apply the control variate adjustment to the estimating equation that defines the Step 1 HERO estimator, with a shared covariate function whose population mean is supplied by the historical sample.  Let $X_i$ denote a covariate observed in the current evaluation task, and let $
\widehat{m}_X = \frac{1}{N} \sum_{r=1}^N \widetilde{X}_r$
denote the sample mean of the same covariate in the historical data. We construct a control-variate estimator based on the score equation in \eqref{eq:score_eq}. After obtaining the pilot HERO estimator $\hat{\mu}_{\textsc{hero-1}}$ from Step 1, define the estimated score contribution for instance $i$ as
\begin{equation}
  \hat{\psi}_i
  :=
  \psi_i(\hat{\mu}_{\textsc{hero-1}})
  =
  \frac{\ell_i(1)-\ell_i(0)}
       {\hat{\mu}_{\textsc{hero-1}}\ell_i(1)+\bigl(1-\hat{\mu}_{\textsc{hero-1}}\bigr)\ell_i(0)}.
  \label{eq:hero1-pilot-score}
\end{equation}
We then estimate the control-variate coefficient by regressing the pilot score $\hat \psi_i$ on $X_i$: 
\begin{equation}
  \hat{\gamma}
  =
  \widehat{\mathrm{Var}}(X)^{-1}
  \widehat{\mathrm{Cov}}(X,\hat{\psi}),
  \label{eq:hero2-gamma}
\end{equation}
where $\widehat{\mathrm{Var}}(X) = \tfrac{1}{n}\sum_{i=1}^n (X_i-\bar X)(X_i-\bar X)^\top$ and $\widehat{\mathrm{Cov}}(X,\hat{\psi}) = \tfrac{1}{n}\sum_{i=1}^n (X_i-\bar X)\hat{\psi}_i$.

Whenever the covariate mean is shared between the historical and current datasets, the difference $\bar X-\widehat{m}_X$ is asymptotically mean zero. Therefore, the adjustment term $\hat{\gamma}^\top(\bar X-\widehat{m}_X)$
can be used to reduce variance without changing the asymptotic expectation of the estimator. We define the control-variate adjusted estimator $\hat{\mu}_{\textsc{hero-2}}$ as the solution to
\begin{equation}
  \frac{1}{n}\sum_{i=1}^n \psi_i(\mu)
  -
  \hat{\gamma}^\top(\bar X-\widehat{m}_X)
  = 0.
  \label{eq:hero2-eq}
\end{equation}
As in Step 1, this equation can also be solved using any standard one-dimensional optimizer.

This adjustment reduces variance whenever the covariate $X_i$ is correlated with the score contribution $\psi_i$, which quantifies how each observation influences the estimator. If $X_i$ carries little information about the outcome, then $\hat{\gamma}$ is close to zero, and the estimator effectively reduces to $\hat{\mu}_{\textsc{hero-1}}$. Consequently, HERO Step 2 improves efficiency when informative auxiliary covariates are available, while leaving the estimator essentially unchanged when they are not.

\subsection{Theoretical Guarantees}\label{subsec:theory}

On top of being practical and easy to implement, HERO enjoys theoretical guarantees along three dimensions: reliability (Theorem 1), sensitivity  (Theorem 2), and valid statistical inference  (Theorem 3). Reliability means that HERO accurately targets the gold-label estimand, rather than a proxy estimand distorted by heterogeneous silver labelers. Sensitivity means that HERO reduces variance, which may improve the detection of small but practically meaningful model-performance differences. Valid statistical inference means that the estimator admits an asymptotic distribution, providing the basis for credible confidence intervals. Due to space limit, we defer the formal theorem statements and proofs to Appendix C.

\begin{figure}[t]
    \centering
    \includegraphics[width=\linewidth]{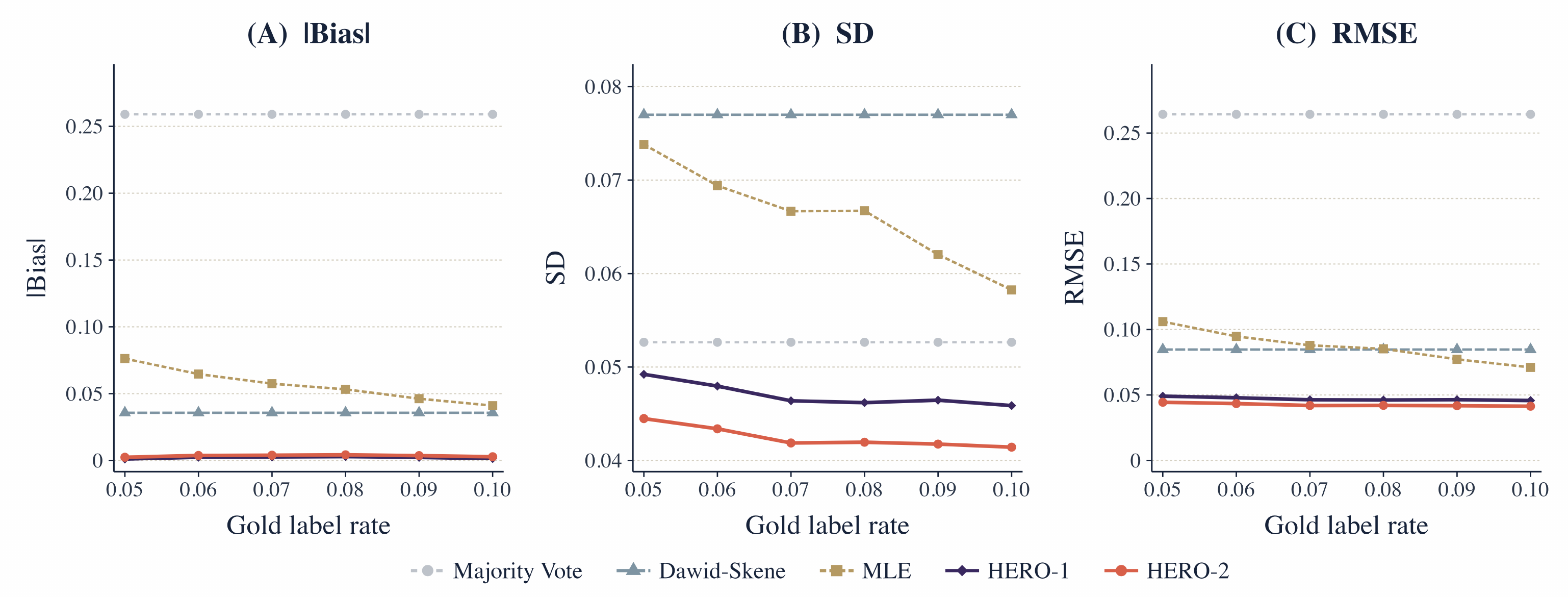}
    \caption{Experiment 1: (A) absolute bias, (B) Monte Carlo standard deviation, and (C) RMSE. }
    \label{fig:sim-gold-large}
\end{figure}

\section{Simulation Study}\label{sec:simulation}

In this section, we evaluate the performance of HERO method
against several benchmark methods. We summarize our takeaway as follows: First, HERO-1 remains accurate even when the historical
gold-label rate is small, whereas the sample-proportion calibration baseline is unstable when only a few
gold-routed examples are available per labeler. This shows HERO-1 can
extract useful reliability information from the full historical label matrix, not only from the gold-labeled
subset. Second, HERO-2 further improves sensitivity by reducing variance and
RMSE relative to HERO-1 when the shared covariate is informative.  Third, under increasing labeler heterogeneity, HERO-1 and HERO-2 remain
substantially more robust than baseline methods, confirming
that historical reliability calibration and score-level control variates reduce estimation errors. 
Due to the page limit, we defer the data generating process details to Appendix. In what follows, we introduce the methods in comparison, evaluation metrics, and simualtion study results.

\textbf{Methods in comparison.}
We compare five estimators: (1) Majority vote. 
(2) Dawid--Skene method with current data only:
The current-only Dawid--Skene baseline fits the standard latent-class model using the current silver labels alone, jointly estimating $\mu$ and the labeler-specific parameters $\{(\alpha_j,\beta_j)\}_{j=1}^{J}$. 
(3) Historical sample-proportion calibration. We estimate each labeler's sensitivity and specificity using only the gold-labeled subset of the historical data. With Laplace smoothing.
We then maximize the current-data likelihood over $\mu$ with these calibration estimates fixed. 
(4) HERO-1 only use historical data for bias reduction. 
(5) HERO-2 applies the score-anchored control-variate adjustment to HERO-1.  In this simulation, the population mean $\mathbb{E}[X]=0$ is known, but we use $\bar X_{\mathrm{hist}}$ to mirror the practical setting in which HERO relies on a high-precision historical estimate rather than oracle knowledge. Appendix Table~\ref{tab:sim-params} summarizes the simulation settings. The target $\mu_0$ is computed from the data-generating model, $\mu_0=\mathbb{E}_{X}\{\sigma(\beta_0+\beta_1X)\}$, and is fixed within each experimental setting.

\textbf{Experiment 1: Historical gold-label rate.}
The first experiment varies the amount of historical gold supervision.  In Setting~(ii), we use a larger labeler pool and a stronger covariate signal: $n=300$, $N=800$, $J=10$, and $\beta_1=3$. We vary $\rho\in\{0.05,0.06,\ldots,0.10\}$. \textbf{Experiment 2: Labeler heterogeneity.} The second experiment studies robustness to heterogeneity in silver-labeler accuracy. We use the same baseline configuration as Experiment~1 Setting~(ii), with $n=300$, $N=800$, $J=10$, $\beta_0=-0.2$, and $\beta_1=3$. We hold the mean reliability fixed at $\mu_\alpha=0.85$ and $\mu_\beta=0.55$, and vary only the Beta concentration parameter $k$. Under the parametrization
$
    \alpha_j\sim \mathrm{Beta}\{k\mu_\alpha,k(1-\mu_\alpha)\},
    \
    \beta_j\sim \mathrm{Beta}\{k\mu_\beta,k(1-\mu_\beta)\},
$
the means satisfy $\mathbb{E}[\alpha_j]=\mu_\alpha$ and $\mathbb{E}[\beta_j]=\mu_\beta$ for every $k$, while
$
    \mathrm{Var}(\alpha_j)=\frac{\mu_\alpha(1-\mu_\alpha)}{k+1},
    \
    \mathrm{Var}(\beta_j)=\frac{\mu_\beta(1-\mu_\beta)}{k+1}.
$
Thus, decreasing $k$ increases labeler heterogeneity without changing the average sensitivity or specificity. We consider the five heterogeneity levels in Appendix Table~\ref{tab:sim-params}. For each level, we evaluate historical gold-label rates $\rho\in\{0.05,0.10\}$.

\begin{figure}
    \centering
    \includegraphics[width=\linewidth]{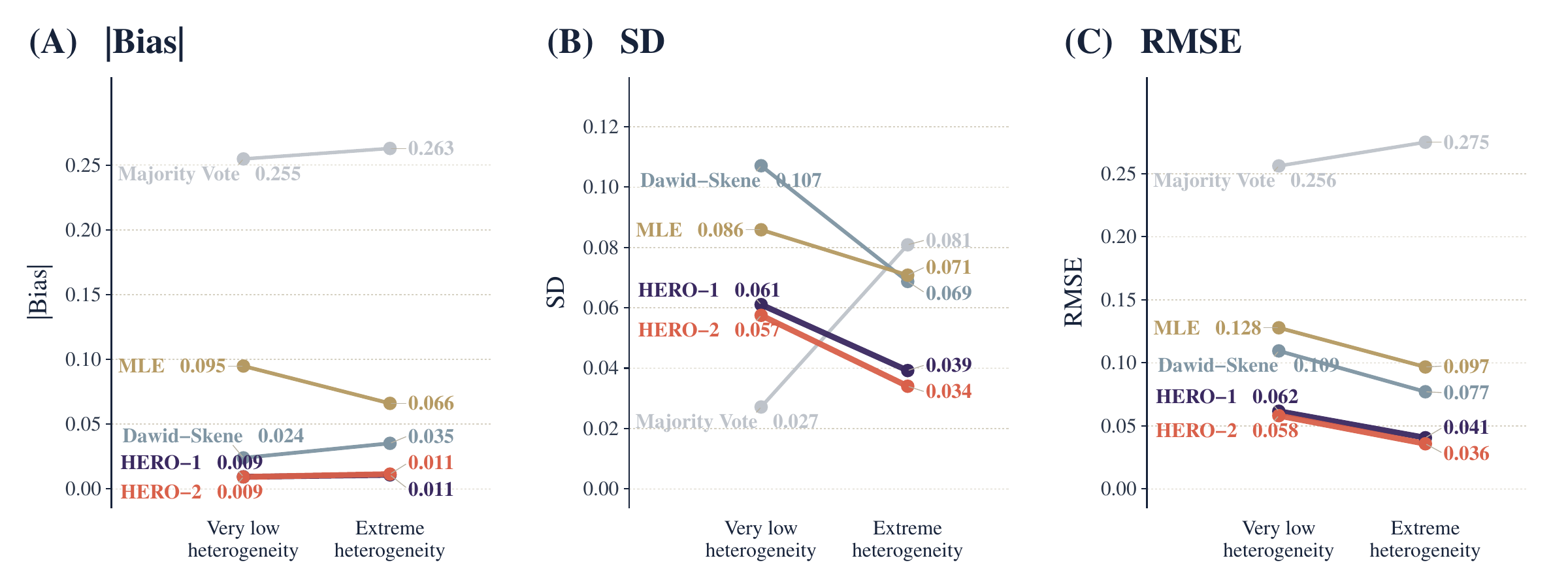}
    \caption{Experiment 2: Comparison of  (A) absolute bias, (B) Monte Carlo standard deviation, and (C) RMSE as the Beta concentration $k$ decreases, increasing the dispersion of labeler-specific sensitivities and specificities }
    \label{fig:sim-heterogeneity}
\end{figure}

\textbf{Evaluation metrics.} For each configuration, we run $B=200$ Monte Carlo replications. We report absolute bias, Monte Carlo standard deviation, and root mean squared error (RMSE). All EM algorithms use convergence tolerance $10^{-6}$ and a maximum of 200 iterations.

\textbf{Simulation study results}.  Figure ~\ref{fig:sim-gold-large} shows that Majority vote has nearly constant bias as $\rho$ varies, since it does not use historical gold labels. The current-only Dawid--Skene baseline also does not improve with $\rho$ and remains limited by the identifiability and finite-sample challenges of estimating all labeler parameters from the small current round alone. The historical sample-proportion baseline improves as $\rho$ increases, but it is unstable when the number of gold-routed examples per labeler is small. In contrast, HERO-1 is substantially less sensitive to low gold-label rates because its anchored EM step uses both the gold-routed and ungolded historical instances. HERO-2 further reduces SD and RMSE by exploiting the correlation between $X_i$ and the score contribution, with larger variance reduction in the stronger-signal setting. We provide additional results in Appendix Figure~\ref{fig:sim-gold-small}.

Figure~\ref{fig:sim-heterogeneity} shows that majority vote becomes increasingly unreliable as labeler heterogeneity grows, because it assigns equal weight to labelers with different error profiles. The current-only Dawid--Skene baseline partially adapts to heterogeneous reliability but remains unstable because the current sample is small and unanchored by gold labels. Historical sample-proportion calibration improves when enough gold labels are available, but it can still have high variance when each labeler receives few gold-routed examples. HERO-1 maintains low bias across heterogeneity levels by pooling information from the full historical label matrix, and HERO-2 achieves the lowest RMSE by combining this calibration with score-level variance reduction.

\section{Results in Real-World Benchmarking Data}\label{sec:real world}

In this section, we evaluate the proposed HERO method on two real-world datasets provided by a technology company that routinely conducts large-scale generative model evaluations. The datasets were collected to assess the performance of 3D asset generation models and contain both historical and current evaluation data. The case studies cover two representative tasks: generation defect evaluation and safety evaluation. In each task, every prompt--asset pair is evaluated by a small number of randomly selected silver labelers according to the task-specific evaluation objective, and a random subset of pairs is additionally evaluated by a gold labeler, who is a well-trained expert. To assess statistical performance and robustness, we use a resampling-based approach to report absolute bias, SD, and RMSE for the five estimators described in Section~\ref{sec:simulation}. Additional details on the real-world datasets are provided in Appendix~E.

\begin{table}[t]
\centering
\begin{tabular}{lcccc}
\toprule
Estimator & Mean estimate & $|\text{Bias}|$ & SD & RMSE \\
\midrule
Majority Vote        & $0.110$ & $0.130$ & $\mathbf{0.019}$ & $0.131$ \\
Dawid--Skene         & $0.120$ & $0.120$ & $0.027$ & $0.123$ \\
MLE (hist.\ sample mean) & $0.069$ & $0.171$ & $0.121$ & $0.209$ \\
HERO-1               & $0.302$ & $0.062$ & $0.070$ & $0.093$ \\
HERO-2               & $0.288$ & $\mathbf{0.048}$ & $0.064$ & $\mathbf{0.080}$ \\
\bottomrule
\end{tabular}
\vspace{0.2cm}
\caption{Case study I (image defect detection), gold-label rate $0.10$.}
\label{tab:defect_gold10}
\end{table}

\begin{figure}
    \centering
    \includegraphics[width=\linewidth]{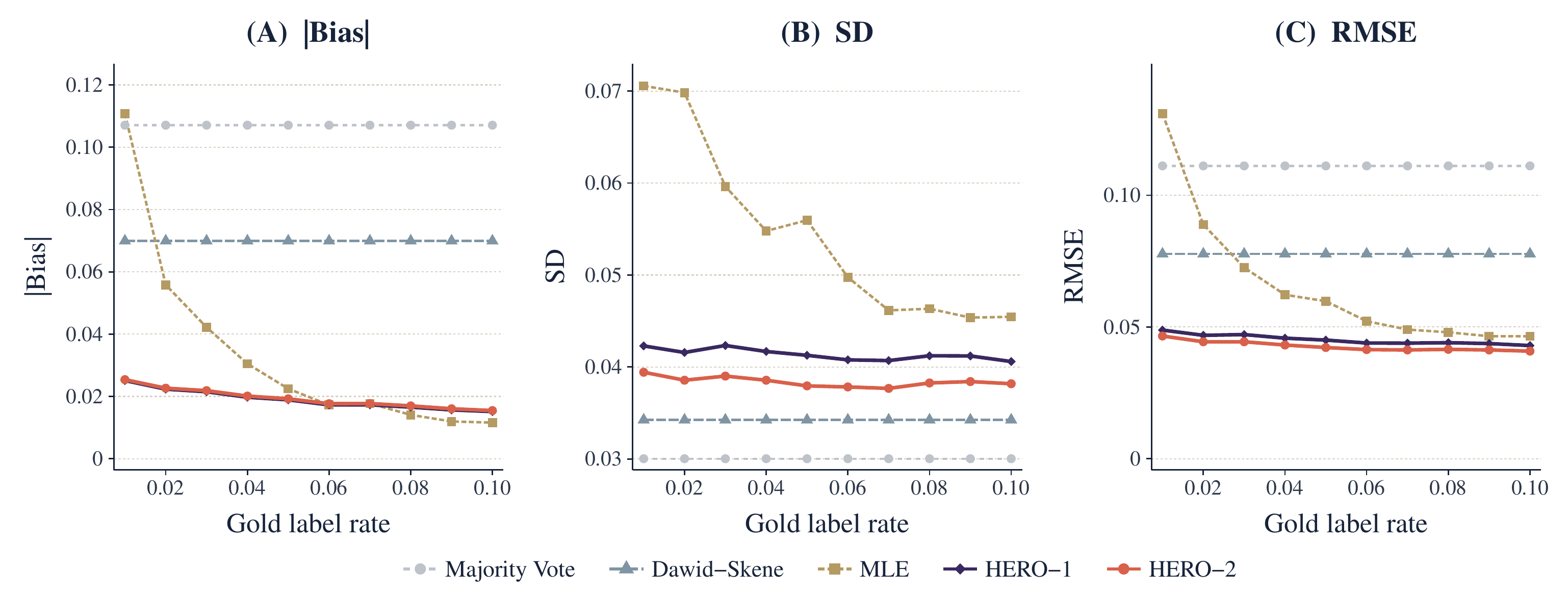}
    \caption{Case study II: Results in model safety evaluation dataset.}
    \label{fig:case_safety_tfidf}
\end{figure}

Table \ref{tab:defect_gold10} and Figure~\ref{fig:case_safety_tfidf} show that HERO method reduces absolute bias and RMSE relative to the baseline methods. This improvement is driven by the presence of labeler heterogeneity and labeler-pool shift between the historical and current datasets, as illustrated in Figure~\ref{fig:bias_sources}. HERO is designed to exploit precisely these sources of information by estimating historical labeler reliabilities and transferring them to the current evaluation task.
Furthermore, HERO-2 achieves additional variance reduction through control variates. In Case Study I, shown in Table \ref{tab:defect_gold10}, HERO-2 uses GPT-5-predicted labels as the control variate, as all generated contents were evaluated by the same model. In Case Study II, shown in Figure~\ref{fig:case_safety_tfidf}, HERO-2 uses TF-IDF features as the control variate. In both cases, HERO-2 reduces the SD of HERO-1 by approximately 5\%--10\%, demonstrating that the proposed variance-reduction framework is effective across different choices of control variates.

In contrast, the classical Dawid--Skene baseline does not effectively learn labeler heterogeneity in either case. As a result, it performs worse than HERO and is often closer to Majority Vote in terms of bias and RMSE.

\textit{Limitation}. Despite these promising real-world results, HERO Step 1 requires labeler reliability to transfer from historical data to the current evaluation task and Step 2 requires the shared covariates to have comparable distributions across the two periods. We view these requirements as mild, but they may not hold uniformly across model evaluation settings.

\clearpage
\bibliographystyle{jasa}
\bibliography{reference}

\clearpage
\appendix
\section{Appendix}

\section{Use cases}

Following the setup in Section 2 in the main manuscript, we provide more details on the four use cases for model evaluations:

\noindent\textit{\textbf{Use Case 1. Single Model Evaluation}.} This scenario focuses on evaluating the performance of a single model (such as defect rate), Model A. We define the outcome as the gold label for the model, 
$Y_i = Y^*_{i,\texttt{A}}.$ 
The target parameter, which represents the average performance or expected quality, is defined as $  \mu_{\texttt{A}} = \mathbb{E}[Y_i]$. 

\noindent\textit{\textbf{Use Case 2. Model Comparative Evaluation}.} This scenario independently evaluates two models, Model A and Model B, on the same prompts to calculate the absolute difference in their performance. We define the outcome as the difference in gold labels between the two models, 
$Y_i = Y^*_{i,\texttt{A}} - Y^*_{i,\texttt{B}}.$
The target parameter for the average difference in performance is expressed as $ \alpha = \mathbb{E}[Y_i]$.

\noindent\textit{\textbf{Use Case 3. Side-by-Side Model Comparison}.} This scenario involves a direct, simultaneous comparison between outputs from multiple models to establish a relative preference. We define the outcome $Y_i$ as a binary indicator of preference rather than an absolute score, such that it equals $1$ if the gold labeler prefers Model A over Model B, and $0$ otherwise. This can be formulated as $Y_i = \mathds{1}(Y^*_{i,\texttt{A}} \succ Y^*_{i,\texttt{B}}).$
The target parameter, representing the win rate or expected preference for Model A, is $\pi_{\texttt{A}\succ\texttt{B}} = \mathbb{E}[Y_i]$.

\noindent\textit{\textbf{Use Case 4. Safety Moderation Calibration}.} This scenario treats the online safety moderation system as an AI automated rater and the offline human review as the gold truth. We let $Y_i^{\texttt{AI}} \in \{0, 1\}$ represent the online model's evaluation, where $Y_i^{\texttt{AI}} = 1$ indicates the content is flagged as unsafe. We let $Y_i^* \in \{0, 1\}$ represent the offline human-labeled ground truth, where $Y_i^* = 1$ indicates an actual safety violation. We define the outcome as the intersection of the gold and AI labels, representing a true positive flag, $ Y_i = Y_i^* Y_i^{\texttt{AI}}.$ The target parameters include the actual violation percentage or prevalence base rate, $p = \mathbb{E}[Y_i^*]$. Additionally, we estimate the online model precision, which is the probability of a true violation given the model flagged it, calculated as $   \text{Precision} = \frac{\mathbb{E}[Y_i]}{\mathbb{E}[Y_i^{\texttt{AI}}]}$. 
The online model recall, which is the probability the model flagged the content given it is a true violation, is defined as $ \text{Recall} = \frac{\mathbb{E}[Y_i]}{\mathbb{E}[Y_i^*]}$.

\section{HERO Algorithmic Presentation}
Please see Algorithm~\ref{alg:hero}.

\begin{algorithm}[h]
\small
\DontPrintSemicolon
\SetAlgoInsideSkip{smallskip}
\setlength{\algomargin}{1.2em}
\caption{HERO}\label{alg:hero}
\KwIn{Historical data $\{(\tilde{Z}_{i}, \tilde{A}_{i}, \tilde{R}_i,
\tilde{Y}_i, \tilde{X}_i)\}_{i=1}^N$;\; Current data $\{(Z_i, A_i, X_i)\}_{i=1}^n$}
\KwOut{$\hat{\mu}_{\text{HERO}}$}
\medskip

\tcc{Stage A. Initialize from the gold-routed historical subset and maximize $\tilde{\mathcal{L}}(\tilde{\mu},\,\{\alpha_j,\beta_j\})$ via EM}
Let $\mathcal{G}=\{i:\tilde{R}_i=1\}$. Initialize, for $j=1,\ldots,m$,
\[
\tilde{\mu}\leftarrow |\mathcal{G}|^{-1}\sum_{i\in\mathcal{G}}\tilde{Y}_i,\qquad
\alpha_j\leftarrow
\frac{\sum_{i\in\mathcal{G}}\tilde{A}_{ij}\tilde{Y}_i\tilde{Z}_{ij}}
     {\sum_{i\in\mathcal{G}}\tilde{A}_{ij}\tilde{Y}_i},
\qquad
\beta_j\leftarrow
\frac{\sum_{i\in\mathcal{G}}\tilde{A}_{ij}(1-\tilde{Y}_i)(1-\tilde{Z}_{ij})}
     {\sum_{i\in\mathcal{G}}\tilde{A}_{ij}(1-\tilde{Y}_i)}.
\]
Clip all initial values to $[\delta,1-\delta]$ for a small $\delta>0$\;
\Repeat{\textnormal{relative change in $\tilde{\mathcal{L}}$} $<\epsilon=o(N^{-1})$}{
  \textbf{E-step:}\;
  \ForEach{$i=1,\dots,N$}{
    \lIf{$\tilde{R}_i=1$}{$\tilde{\gamma}_i\leftarrow\tilde{Y}_i$}
    \lElse{$\tilde{\gamma}_i\leftarrow
      \tilde{\mu}\,\tilde{\ell}_i(1)\big/\bigl[
      \tilde{\mu}\,\tilde{\ell}_i(1)+(1-\tilde{\mu})\,\tilde{\ell}_i(0)\bigr]$}
  }
  \textbf{M-step:}\;
  $\tilde{\mu}\leftarrow N^{-1}\textstyle\sum_{i=1}^N\tilde{\gamma}_i$\;
  \ForEach{$j=1,\dots,m$}{
  \[
  \alpha_j\leftarrow
    \frac{\sum_i\tilde{A}_{ij}\tilde{\gamma}_i\tilde{Z}_{ij}}
         {\sum_i\tilde{A}_{ij}\tilde{\gamma}_i},
  \qquad
  \beta_j\leftarrow
    \frac{\sum_i\tilde{A}_{ij}(1-\tilde{\gamma}_i)(1-\tilde{Z}_{ij})}
         {\sum_i\tilde{A}_{ij}(1-\tilde{\gamma}_i)}.
  \]
  Clip $\alpha_j,\beta_j$ to $[\delta,1-\delta]$\;
  }
  Clip $\tilde{\mu}$ to $[\delta,1-\delta]$\;
}
Store $(\hat{\alpha}_j,\hat{\beta}_j)_{j=1}^m$\;

\medskip
\tcc{Stage B. Maximize $L(\mu)$ with $(\hat{\alpha}_j,\hat{\beta}_j)$}
\ForEach{$i=1,\dots,n$}{
  $\ell_i(1)\leftarrow\prod_{j:A_{ij}=1}
    \hat{\alpha}_j^{Z_{ij}}(1-\hat{\alpha}_j)^{1-Z_{ij}}$,\quad
  $\ell_i(0)\leftarrow\prod_{j:A_{ij}=1}
    (1-\hat{\beta}_j)^{Z_{ij}}\hat{\beta}_j^{1-Z_{ij}}$\;
}
$\hat{\mu}_{\text{HERO-1}}\leftarrow
  \arg\max_{\mu\in(0,1)}\sum_{i=1}^n
  \log\bigl[\mu\,\ell_i(1)+(1-\mu)\,\ell_i(0)\bigr]$\;

\medskip
\tcc{Stage C. Variance reduction via control variates}
\lIf{no shared covariates available}{\Return $\hat{\mu}_{\text{HERO-1}}$}
$\hat{m}_X \leftarrow N^{-1}\sum_{r=1}^{N}\tilde{X}_r$,\quad
$\bar{X}\leftarrow n^{-1}\sum_{i=1}^{n}X_i$\;
\ForEach{$i=1,\dots,n$}{
  $\hat{\psi}_i\leftarrow
    \frac{\ell_i(1)-\ell_i(0)}
         {\hat{\mu}_{\text{HERO-1}}\,\ell_i(1)
          +(1-\hat{\mu}_{\text{HERO-1}})\,\ell_i(0)}$\;
}
$\hat{\gamma}\leftarrow
  \Bigl[\frac{1}{n}\sum_{i=1}^{n}(X_i-\bar{X})(X_i-\bar{X})^\top\Bigr]^{-1}
  \Bigl[\frac{1}{n}\sum_{i=1}^{n}(X_i-\bar{X})\hat{\psi}_i\Bigr]$\;
$\hat{\mu}_{\text{HERO-2}}\leftarrow
  \text{Find in }(0,1)\text{ such that}$
\[
  \frac{1}{n}\sum_{i=1}^{n}
  \frac{\ell_i(1)-\ell_i(0)}
       {\mu\,\ell_i(1)+(1-\mu)\,\ell_i(0)}
  -
  \hat{\gamma}^{\top}(\bar{X}-\hat{m}_X)
  =0
\]
\Return $\hat{\mu}_{\text{HERO-2}}$
\end{algorithm}

\section{Theoretical Guarantees of HERO}

This section establishes the theoretical guarantees of HERO. The theoretical results rely on three assumptions on the labeling mechanism, the transportability of labeler reliability between historical and current data, and the sampling structure. Lemma~1 shows that the reliability estimates $\widehat{\alpha}_j$ and $\widehat{\beta}_j$ converge at rate $O_p(N^{-1/2})$ when historical gold labels anchor the EM, and Lemma~2 shows that the same EM algorithm without a gold anchor may fail to recover the true labeler parameters. Three theorems then state the main results. Theorem~1 shows that HERO-1 attains the parametric rate $O_p(n^{-1/2})$ for the gold-label target, while the majority-vote and unanchored Dawid--Skene baselines are not guaranteed to converge to that target. Theorem~2 shows that HERO-2 reduces the asymptotic variance of HERO-1 whenever the shared covariate correlates with the score function, and preserves the bias of HERO-1 up to higher-order terms. Theorem~3 establishes asymptotic normality for both estimators, providing the basis for standard error estimation and Wald-type confidence intervals.

\begin{assumption}[Randomized labeling]
\label{assump:random-labeling}
The silver-label availability indicators $A_{ij}$ and $\tilde{A}_{ij}$ are completely randomized, satisfying
\[
P(A_{ij} = 1) \ge \pi_{\min} > 0, \quad P(\tilde{A}_{ij} = 1) \ge \pi_{\min} > 0,
\]
and the historical gold-routing indicators $\tilde{R}_i$ are completely randomized, satisfying
\[
P(\tilde{R}_i = 1) = \rho \ge \rho_{\min} > 0,
\]
where $\pi_{\min}$ and $\rho_{\min}$ are positive constants.
\end{assumption}

\begin{assumption}[Correct specification and transportable reliability]
\label{ass:transportable-reliability}
For each labeler $j$, there exist constants $(\alpha_j,\beta_j)$ satisfying $c < \alpha_j, \beta_j < 1-c$ for some $c>0$ such that the conditional reliabilities are identical in historical and current data:
\begin{align*}
\Pr(\widetilde Z_{ij}=1 \mid \widetilde Y_i=1) = \Pr(Z_{ij}=1 \mid Y_i=1) = \alpha_j,
\\
\Pr(\widetilde Z_{ij}=0 \mid \widetilde Y_i=0) = \Pr(Z_{ij}=0 \mid Y_i=0) = \beta_j,  
\end{align*}
for all instances $i$ and all labelers $j$.  
The silver labels are independent across labelers conditional on the gold label:
\begin{align}
\Pr\big(\widetilde Z_{i1} = z_1, \dots, \widetilde Z_{iJ} = z_J \;\big|\; \widetilde Y_i = y\big)
&= 
\prod_{j=1}^m \Pr(\widetilde Z_{ij} = z_j \mid \widetilde Y_i = y), \\
\Pr\big(Z_{i1} = z_1, \dots, Z_{iJ} = z_J \;\big|\; Y_i = y\big)
&=
\prod_{j=1}^m \Pr(Z_{ij} = z_j \mid Y_i = y),
\end{align}
for all instances $i$ and labelers $j=1,\dots,m$.
The parameters $\mu = \mathbb E[Y]$ and $\tilde\mu = \mathbb E[\widetilde Y]$ may differ, but are bounded away from $0$ and $1$:
\[
c < \mu, \tilde\mu < 1-c.
\]
\end{assumption}

\begin{lemma}[Convergence of HERO historical reliability estimates]\label{lem:em convergence}
Under Assumptions \ref{assump:random-labeling} and \ref{ass:transportable-reliability}, the estimators for $\alpha_j$ and $\beta_j$ in HERO algorithm satisfy
\[
\hat{\alpha}_j - \alpha_j = O_p(N^{-1/2}), \qquad 
\hat{\beta}_j - \beta_j = O_p(N^{-1/2}), \quad 
\text{for each labeler } j.
\]
\end{lemma}

\begin{proof}
    
Let $\mathcal{G} = \{ i : \tilde{R}_i = 1 \}$ denote the set of gold-routed historical instances. By the algorithm, the EM is initialized using the gold-only sample-mean estimators
\[
\tilde{\mu}^{(0)} = \frac{1}{|\mathcal{G}|} \sum_{i \in \mathcal{G}} \tilde{Y}_i, \quad
\hat{\alpha}_j^{(0)} = \frac{\sum_{i \in \mathcal{G}} \tilde{A}_{ij} \tilde{Y}_i \tilde{Z}_{ij}}{\sum_{i \in \mathcal{G}} \tilde{A}_{ij} \tilde{Y}_i}, \quad 
\hat{\beta}_j^{(0)} = \frac{\sum_{i \in \mathcal{G}} \tilde{A}_{ij} (1-\tilde{Y}_i)(1-\tilde{Z}_{ij})}{\sum_{i \in \mathcal{G}} \tilde{A}_{ij} (1-\tilde{Y}_i)}, \quad j=1,\dots,m.
\]
Under Assumption \ref{assump:random-labeling}, the gold-routing indicators $\tilde{R}_i$ are completely randomized with $P(\tilde{R}_i = 1) = \rho \ge \rho_{\min} > 0$ and the silver-label availability indicators $\tilde{A}_{ij}$ are completely randomized with $P(\tilde{A}_{ij}=1) \ge \pi_{\min}>0$, so by the law of large numbers and the delta method, for each $j$,
\[
\hat{\alpha}_j^{(0)} - \alpha_j = O_p\Big((N \rho \pi_j)^{-1/2}\Big), \qquad \hat{\beta}_j^{(0)} - \beta_j = O_p\Big((N \rho \pi_j)^{-1/2}\Big), \qquad \tilde{\mu}^{(0)} - \tilde{\mu} = O_p\Big((N \rho)^{-1/2}\Big),
\]
which implies that the initialization lies within a neighborhood of the gold-oriented population parameter $\eta_0 = (\tilde{\mu}, \alpha_1, \beta_1, \dots, \alpha_m, \beta_m)$. Under Assumption \ref{ass:transportable-reliability}, the historical data follow the Dawid--Skene model with true gold-oriented parameters $\eta_0 = (\tilde{\mu}, \alpha_1, \beta_1, \dots, \alpha_m, \beta_m)$, so that for each $i$ and $j$ the conditional distributions of $\tilde{Z}_{ij}$ given $\tilde{Y}_i$ are correctly specified and independent across labelers. The observed-data log-likelihood for the historical dataset is
\[
\tilde{\mathcal{L}}_N(\tilde{\mu}, \{\alpha_j, \beta_j\}) = \sum_{i=1}^N \Big[ \tilde{R}_i \log \tilde{\ell}_i(\tilde{Y}_i) + (1-\tilde{R}_i) \log \big( \tilde{\mu} \tilde{\ell}_i(1) + (1-\tilde{\mu}) \tilde{\ell}_i(0) \big) \Big].
\]
By construction, this likelihood is twice continuously differentiable in a neighborhood of $\eta_0$, and the Fisher information matrix
\[
\mathcal{I}_\eta = - E\Big[ \frac{\partial^2}{\partial \eta \partial \eta^\top} \log P(\tilde{Z}_i \mid \tilde{Y}_i) \Big|_{\eta = \eta_0} \Big]
\]
is nonsingular because each labeler contributes a positive amount of information and $\alpha_j, \beta_j \in (c, 1-c)$ for some $c>0$, ensuring strong identifiability. By standard results on EM for latent-variable models \citep{dempster1977maximum,wu1983convergence}, the EM map is monotone in the log-likelihood, i.e.,
\[
\tilde{\mathcal{L}}_N(\eta^{(t+1)}) \ge \tilde{\mathcal{L}}_N(\eta^{(t)}),
\]
and converges to a stationary point of $\tilde{\mathcal{L}}_N$. Since the initialization $\eta^{(0)}$ is obtained from the gold-only sample mean and satisfies $\|\eta^{(0)} - \eta_0\| = O_p((N\rho \pi_{\min})^{-1/2})$, it lies within the basin of attraction of the gold-oriented local maximizer $\widehat{\eta}_N$, and therefore the EM iterates converge to $\widehat{\eta}_N$. By the stopping criterion $\epsilon = o(N^{-1})$, the EM optimization error satisfies
\[
\|\eta^{(T)} - \widehat{\eta}_N\| = o_p(N^{-1/2}),
\]
so that it is negligible relative to the statistical variability of the maximizer. Finally, applying standard M-estimation theory for local maximizers of smooth likelihoods \citep{van2000asymptotic}, we have
\[
\widehat{\eta}_N - \eta_0 = O_p(N^{-1/2}),
\]
and combining with the EM optimization error yields
\[
\hat{\alpha}_j - \alpha_j = O_p(N^{-1/2}), \quad \hat{\beta}_j - \beta_j = O_p(N^{-1/2}), \quad \hat{\mu} - \tilde{\mu} = O_p(N^{-1/2}), \qquad j=1,\dots,m.
\]
\end{proof}

\begin{lemma}[EM without gold anchor may fail]
\label{lem:DS-no-gold-short}
Consider the Dawid--Skene EM applied to historical data or current data without using any gold labels. Then the EM does not necessarily recover the true parameters $\alpha$ and $\beta$.
\end{lemma}

\begin{proof}
Let the observed-data likelihood for the unsupervised DS problem be
\[
\tilde{\mathcal{L}}_N(\mu, \{\alpha_j,\beta_j\}) = \sum_{i=1}^N \log \bigl( \mu \tilde{\ell}_i(1) + (1-\mu) \tilde{\ell}_i(0) \bigr),
\]
with
\[
\tilde{\ell}_i(1) = \prod_{j:\tilde{A}_{ij}=1} \alpha_j^{\tilde{Z}_{ij}} (1-\alpha_j)^{1-\tilde{Z}_{ij}}, \quad
\tilde{\ell}_i(0) = \prod_{j:\tilde{A}_{ij}=1} (1-\beta_j)^{\tilde{Z}_{ij}} \beta_j^{1-\tilde{Z}_{ij}}.
\]

Consider the second derivative of the likelihood with respect to any \(\alpha_j\):
\[
\frac{\partial^2}{\partial \alpha_j^2} \log \bigl( \mu \tilde{\ell}_i(1) + (1-\mu) \tilde{\ell}_i(0) \bigr)
= \frac{-\mu^2 \bigl(\partial_{\alpha_j} \tilde{\ell}_i(1)\bigr)^2 + \mu \bigl(\mu \partial_{\alpha_j}^2 \tilde{\ell}_i(1) + (1-\mu)\partial_{\alpha_j}^2 \tilde{\ell}_i(0)\bigr) \tilde{\ell}_i(1)}{(\mu \tilde{\ell}_i(1) + (1-\mu) \tilde{\ell}_i(0))^2},
\]
which depends on both \(\tilde{\ell}_i(1)\) and \(\tilde{\ell}_i(0)\). Since \(\tilde{\ell}_i(0)\) depends nonlinearly on \(\beta_j\) and \(\tilde{\ell}_i(1)\) depends on \(\alpha_j\), the Hessian is not negative semidefinite in general. Therefore, \(\tilde{\mathcal{L}}_N\) is not concave in $(\mu, \alpha_j, \beta_j)$.  

Because EM is a local algorithm, it is only guaranteed to converge to a stationary point of \(\tilde{\mathcal{L}}_N\) \citep{dempster1977maximum,wu1983convergence}. Hence, EM without a gold anchor does not guarantee consistent recovery of the true maxima.
\end{proof}

\begin{theorem}[Estimation rate and comparison with baselines]
\label{cor:mu-rate}
Under Assumptions \ref{assump:random-labeling} and \ref{ass:transportable-reliability}, the HERO-1 estimator satisfies
\[
\hat{\mu}_{\text{HERO-1}} - \mu = O_p(n^{-1/2}) + O_p(N^{-1/2}).
\]
In contrast, the Majority Vote estimator
and the classical Dawid--Skene EM estimator without gold anchor do not necessarily converge to the gold-oriented truth.
\end{theorem}

\begin{proof}
Decompose the HERO-1 estimation error as
\[
\hat{\mu}_{\text{HERO-1}} - \mu = (\hat{\mu}_{\text{HERO-1}} - \mu^*) + (\mu^* - \mu),
\]
where 
\[
\mu^* = \arg\max_{\mu\in(0,1)} \sum_{i=1}^n \log \bigl[ \mu \ell_i(1;\alpha_j,\beta_j) + (1-\mu) \ell_i(0;\alpha_j,\beta_j) \bigr]
\]
is the oracle MLE computed using the true labeler parameters $(\alpha_j,\beta_j)$.

By standard MLE theory \citep{van2000asymptotic}, the oracle MLE satisfies
\[
\mu^* - \mu = O_p(n^{-1/2}),
\]
because the log-likelihood for the current data is smooth and concave in $\mu$, and the $n$ instances are independent.

From Lemma~\ref{lem:em convergence}, the historical HERO EM estimates satisfy
\[
\hat{\alpha}_j - \alpha_j = O_p(N^{-1/2}), \qquad \hat{\beta}_j - \beta_j = O_p(N^{-1/2}), \quad j=1,\dots,m.
\]
Since the current-data likelihood is smooth in $(\alpha_j, \beta_j)$, a first-order Taylor expansion gives
\[
\ell_i(1;\hat{\alpha}_j) - \ell_i(1;\alpha_j) = O_p(N^{-1/2}), \qquad
\ell_i(0;\hat{\beta}_j) - \ell_i(0;\beta_j) = O_p(N^{-1/2}),
\]
uniformly in $i$. Applying the delta method to the MLE equation for $\hat{\mu}_{\text{HERO-1}}$, we obtain
\[
\hat{\mu}_{\text{HERO-1}} - \mu^* = O_p(N^{-1/2}).
\]

Combining the two contributions via the triangle inequality yields
\[
\hat{\mu}_{\text{HERO-1}} - \mu = O_p(n^{-1/2}) + O_p(N^{-1/2}).
\]
\end{proof}

\begin{assumption}[Instance sampling and moments]
\label{assump:hero2}
The current instances $\{i=1,\dots,n\}$ and historical instances $\{i=1,\dots,N\}$ are independently and identically distributed from the same population, with sample size $N \gg n$. 
The HERO-1 score contributions $\psi_i(\mu_0)$ and covariates $X_i$ have finite second moments:
\[
\mathbb{E}[\psi_i(\mu_0)^2] < \infty, \quad \mathbb{E}[X_i^2] < \infty.
\]
\end{assumption}

\begin{theorem}[Variance reduction of HERO-2 over HERO-1]
\label{thm:hero2-var-compact}
Under Assumption~\ref{assump:hero2}, the asymptotic variances  of the HERO-1 and HERO-2 estimators satisfy
\[ \mathbb V(\hat\mu_{\textsc{hero-2}}) = \mathbb V(\hat\mu_{\textsc{hero-1}}) - \mathcal I_0^{-2} \mathbb C\mathrm{ov}\{\psi_i(\mu_0),X_i\}^{\top} \Sigma_X^{-1} \mathbb C\mathrm{ov}\{X_i,\psi_i(\mu_0)\}, \] where $ \mathcal I_0=-\partial_\mu\mathbb E[\psi_i(\mu)]\big|_{\mu=\mu_0}>0$. Consequently, \[ \mathbb V(\hat\mu_{\textsc{hero-2}}) \le \mathbb V(\hat\mu_{\textsc{hero-1}}), \] with strict reduction whenever $\mathbb C\mathrm{ov}\{X_i,\psi_i(\mu_0)\}\neq 0$. Moreover, HERO-2 preserves the bias reduction of HERO-1: \[ \operatorname{Bias}(\hat\mu_{\textsc{hero-2}}) = \operatorname{Bias}(\hat\mu_{\textsc{hero-1}}) + O(N^{-1/2})+o(n^{-1/2}). \] \end{theorem}

\begin{proof}
We first consider the oracle HERO-2 estimator using the true control-variate coefficient $\gamma^*$ and the population mean $m_X$. 
Linearizing the score equation around the true parameter $\mu_0$, we have
\[
\hat\mu_{\textsc{hero-2}} - \mu_0 
= I_0^{-1} \frac{1}{n} \sum_{i=1}^n \bigl(\psi_i(\mu_0) - \gamma^{*\top} (X_i - m_X)\bigr) + R_n,
\]
where $I_0 = -\partial_\mu \mathbb E[\psi_i(\mu)]\big|_{\mu=\mu_0} > 0$ and $R_n = o_p(n^{-1/2})$ is the remainder from the Taylor expansion of the score function.  
By the standard projection property of linear regression, the variance decomposes as
\[
\operatorname{Var}\bigl[\psi_i(\mu_0) - \gamma^{*\top} (X_i - m_X)\bigr] 
= \operatorname{Var}[\psi_i(\mu_0)] 
- \operatorname{Cov}(\psi_i(\mu_0), X_i)^\top \Sigma_X^{-1} \operatorname{Cov}(X_i, \psi_i(\mu_0)) 
\le \operatorname{Var}[\psi_i(\mu_0)],
\]
with strict inequality whenever $\operatorname{Cov}(\psi_i(\mu_0), X_i) \neq 0$.  

Using Assumption~\ref{assump:hero2}, the historical mean $\hat m_X$ satisfies $\hat m_X - m_X = O_p(N^{-1/2})$ and the estimated coefficient $\hat\gamma$ satisfies $\hat\gamma - \gamma^* = o_p(1)$. These contribute only a higher-order term of 
\[
O_p(N^{-1/2}) + o_p(n^{-1/2})
\]
to the asymptotic expansion of $\hat\mu_{\textsc{hero-2}}$ since $N\gg n$, which does not affect the first-order variance reduction.  

Therefore, replacing $m_X$ by $\hat m_X$ and $\gamma^*$ by $\hat\gamma$ preserves the variance reduction, and the bias satisfies
\[
\operatorname{Bias}(\hat\mu_{\textsc{hero-2}}) 
= \operatorname{Bias}(\hat\mu_{\textsc{hero-1}}) + O(N^{-1/2}) + o(n^{-1/2}).
\]

This proves that HERO-2 strictly reduces variance relative to HERO-1 while preserving its bias properties.
\end{proof}

\begin{theorem}[Asymptotic normality of HERO-1 and HERO-2]
Suppose Assumptions~\ref{assump:random-labeling}--\ref{assump:hero2} hold, then HERO-1 satisfies
\[
\sqrt n(\hat\mu_{\textsc{hero-1}}-\mu_0)
\xrightarrow{d}
\mathcal{N}\left(
0,\,
\mathcal I_0^{-2}\mathbb V\mathrm{ar}\{\psi_i(\mu_0)\}
\right),
\]
and HERO-2 satisfies
\[
\sqrt n(\hat\mu_{\textsc{hero-2}}-\mu_0)
\rightsquigarrow
N\left(
0,\,
\mathcal I_0^{-2}
\mathbb V\mathrm{ar}\{\psi_i(\mu_0)-\gamma^{*\top}(X_i-m_X)\}
\right).
\]

\end{theorem}

\begin{proof}
For HERO-1, consider first the idealized MLE using the true labeler reliabilities $\alpha_j$ and $\beta_j$. The score equation for $\mu$ is smooth and strictly concave, so by the classical M-estimator theory, we have
\[
\sqrt n \big(\hat\mu_{\textsc{hero-1}}(\alpha,\beta) - \mu_0\big) 
\xrightarrow{d} 
\mathcal{N}\Big(0, \mathcal I_0^{-2} \operatorname{Var}\{\psi_i(\mu_0)\}\Big),
\]
where $\psi_i(\mu_0)$ is the per-instance score and $\mathcal I_0 = -\partial_\mu \mathbb{E}[\psi_i(\mu)]|_{\mu=\mu_0} >0$.  

By Lemma~\ref{lem:em convergence}, the estimated $\hat\alpha_j, \hat\beta_j$ from the historical data satisfy $\hat\alpha_j-\alpha_j = O_p(N^{-1/2})$, $\hat\beta_j-\beta_j = O_p(N^{-1/2})$. Using Assumption~\ref{assump:hero2} that $N \gg n$, the contribution of the estimation error to the $\mu$-MLE is
\[
I_0^{-1} \frac{\partial \psi_i}{\partial (\alpha,\beta)} (\hat\alpha - \alpha, \hat\beta - \beta) = o_p(n^{-1/2}),
\]
so $\sqrt n (\hat\mu_{\textsc{hero-1}} - \mu_0)$ has the same asymptotic distribution as the idealized MLE.  

For HERO-2, consider first the oracle estimator using the true control-variate coefficient $\gamma^*$ and population mean $m_X$. Linearizing the adjusted score equation, we have
\[
\sqrt n \big(\hat\mu_{\textsc{hero-2}}(\gamma^*, m_X) - \mu_0\big)
= \mathcal I_0^{-1} \frac{1}{\sqrt n} \sum_{i=1}^n \big(\psi_i(\mu_0) - \gamma^{*\top} (X_i - m_X)\big) + o_p(1),
\]
which immediately yields the asymptotic normality with variance 
\(\mathcal I_0^{-2} \operatorname{Var}[\psi_i(\mu_0)-\gamma^{*\top}(X_i - m_X)]\).  

Finally, by the argument in the proof of Theorem~\ref{thm:hero2-var-compact}, replacing $\gamma^*$ and $m_X$ by their estimators $\hat\gamma$ and $\hat m_X$ contributes only $o_p(1)$ to the $\sqrt n$-scale expansion. Hence the asymptotic distribution of HERO-2 is unchanged. This completes the proof.
\end{proof}

\section{Additional simulation details and results}

In this section, we shall provide additional simulation setup details in Section \ref{subsec:simulation setup additional details} and additional simulation results in Section \ref{subsec:Additional simulation results}.

\subsection{Simulation setup additional details.}
\label{subsec:simulation setup additional details}

In Table \ref{tab:heterogeneity-levels}, we shall provide the details of how we define heterogeneity level in Experiment 2. In Table \ref{tab:sim-params}, we summarize the simulation parameters used in the main manuscript.  In what follows, we shall describe the data generating process. In each Monte Carlo replication, we generate two datasets: a large historical evaluation sample of size $N$ and a smaller current evaluation sample of size $n$. The two samples share the same prompt-feature distribution and the same pool of silver labelers, but gold labels are observed only for a randomized subset of historical instances. This construction matches the intended deployment setting of HERO, where historical gold labels are used to learn persistent labeler reliability and historical covariates provide a high-precision control-variate anchor. For each instance $i$, we draw a scalar covariate $
    X_i \sim \mathcal{N}(0,1),
$
and generate the latent gold label from the logistic model
$
    Y_i \mid X_i \sim \mathrm{Bernoulli}\{p(X_i)\},$
$ p(x) = \sigma(\beta_0+\beta_1 x)$, $\sigma(t)=\frac{1}{1+\exp(-t)} .$
The estimand is therefore $
    \mu_0 := \mathbb{E}_{X}\{\sigma(\beta_0+\beta_1 X)\}.$
We set $\beta_0=-0.2$. The slope $\beta_1$ controls the correlation between the auxiliary covariate $X_i$ and the latent label $Y_i$: larger $\beta_1$ makes $X_i$ more informative and therefore increases the potential gain from the HERO-2 control-variate adjustment.

Each silver labeler $j\in\{1,\ldots,J\}$ has a sensitivity $\alpha_j$ and specificity $\beta_j$,
$
    \alpha_j = \Pr(Z_{ij}=1\mid Y_i=1),
    \beta_j = \Pr(Z_{ij}=0\mid Y_i=0).$
To induce persistent labeler heterogeneity, we draw
$
    \alpha_j \overset{\mathrm{iid}}{\sim}
    \mathrm{Beta}\{k\mu_\alpha,k(1-\mu_\alpha)\},
    \beta_j \overset{\mathrm{iid}}{\sim}
    \mathrm{Beta}\{k\mu_\beta,k(1-\mu_\beta)\}.
$
Throughout Experiment~1 we use $\mu_\alpha=0.85$, $\mu_\beta=0.55$, and $k=8$. Thus, the average silver labeler has high sensitivity but only moderate specificity, creating an asymmetric error pattern under which uncalibrated majority vote is expected to be biased. The concentration parameter $k$ controls across-labeler heterogeneity: larger $k$ yields a more homogeneous labeler pool around the same mean reliability, whereas smaller $k$ yields more dispersed labeler-specific error rates. Each labeler $j$ annotates instance $i$ with probability $
    A_{ij}\sim \mathrm{Bernoulli}(\pi_j),\ 
    \pi_j\sim \mathrm{Uniform}(0.2,0.6),$
independently across instances conditional on $\pi_j$. If $A_{ij}=1$, the observed silver label is generated as
$Z_{ij}\mid Y_i =
    \mathrm{Bernoulli}(\alpha_j),$ if $ Y_i=1$, and $  Z_{ij}\mid Y_i =
    \mathrm{Bernoulli}(1-\beta_j)$, if $Y_i=0$.
Instances with no observed silver labels are excluded from all estimators; this event is rare in the reported settings and affects all methods equally. The historical and current datasets are generated from the same mechanism, with independent instances but a common labeler pool. For the historical sample, a randomized routing indicator
$
    R_i \sim \mathrm{Bernoulli}(\rho)
$
determines whether the latent gold label is revealed. The parameter $\rho$ is the historical gold-label rate. For the current sample, $R_i=0$ for all instances, so no current gold labels are observed.

\begin{table}[ht]
\centering
\caption{Heterogeneity levels used in Experiment~2. Smaller $k$ corresponds to greater across-labeler dispersion while preserving the same mean sensitivity and specificity.}
\label{tab:heterogeneity-levels}
\begin{tabular}{lccccc}
\toprule
Category & $k$ & $\mathbb{E}[\alpha_j]$ & $\mathrm{SD}(\alpha_j)$
& $\mathbb{E}[\beta_j]$ & $\mathrm{SD}(\beta_j)$ \\
\midrule
Very low & $150$ & $0.85$ & $0.029$ & $0.55$ & $0.040$ \\
Low      & $30$  & $0.85$ & $0.064$ & $0.55$ & $0.089$ \\
Moderate & $12$  & $0.85$ & $0.099$ & $0.55$ & $0.138$ \\
High     & $6$   & $0.85$ & $0.135$ & $0.55$ & $0.188$ \\
Extreme  & $3$   & $0.85$ & $0.179$ & $0.55$ & $0.249$ \\
\bottomrule
\end{tabular}
\end{table}

\begin{table}[ht]
\centering
\caption{Simulation parameters.}{Simulation parameters used in the Monte Carlo experiments.}
\label{tab:sim-params}
\begin{tabular}{lll}
\toprule
Parameter & Value & Description \\
\midrule
$n$ & $150$ & Current instances in Experiment~1 Setting~(i) \\
$n$ & $300$ & Current instances in Experiment~1 Setting~(ii) and Experiment~2 \\
$N$ & $400$ & Historical instances in Experiment~1 Setting~(i) \\
$N$ & $800$ & Historical instances in Experiment~1 Setting~(ii) and Experiment~2 \\
$J$ & $5$ & Number of labelers in Experiment~1 Setting~(i) \\
$J$ & $10$ & Number of labelers in Experiment~1 Setting~(ii) and Experiment~2 \\
$\beta_0$ & $-0.2$ & Logistic intercept \\
$\beta_1$ & $1.0$ & Logistic slope in Experiment~1 Setting~(i) \\
$\beta_1$ & $3.0$ & Logistic slope in Experiment~1 Setting~(ii) and Experiment~2 \\
$\mu_\alpha$ & $0.85$ & Mean silver-labeler sensitivity \\
$\mu_\beta$ & $0.55$ & Mean silver-labeler specificity \\
$k$ & $8$ & Baseline Beta concentration in Experiment~1 \\
$\pi_j$ & $\mathrm{Uniform}(0.2,0.6)$ & Labeler-specific annotation probability \\
$\rho$ & $\{0.01,0.02,\ldots,0.10\}$ & Historical gold-label rate in Experiment~1 Setting~(i) \\
$\rho$ & $\{0.05,0.06,\ldots,0.10\}$ & Historical gold-label rate in Experiment~1 Setting~(ii) \\
$k$ & $\{150,30,12,6,3\}$ & Beta concentration levels in Experiment~2 \\
$\rho$ & $\{0.05,0.10\}$ & Historical gold-label rates in Experiment~2 \\
$B$ & $200$ & Monte Carlo replications \\
\bottomrule
\end{tabular}
\end{table}

\subsection{Additional simulation results}\label{subsec:Additional simulation results}

In what follows, we show another setting of Experiment 1. In this setting, we use a small current sample and a moderate covariate signal: $n=150$, $N=400$, $J=5$, and $\beta_1=1$. We vary the historical gold-label rate over $\rho\in\{0.01,0.02,\ldots,0.10\}$.

\begin{figure}[ht]
    \centering  \includegraphics[width=\linewidth]{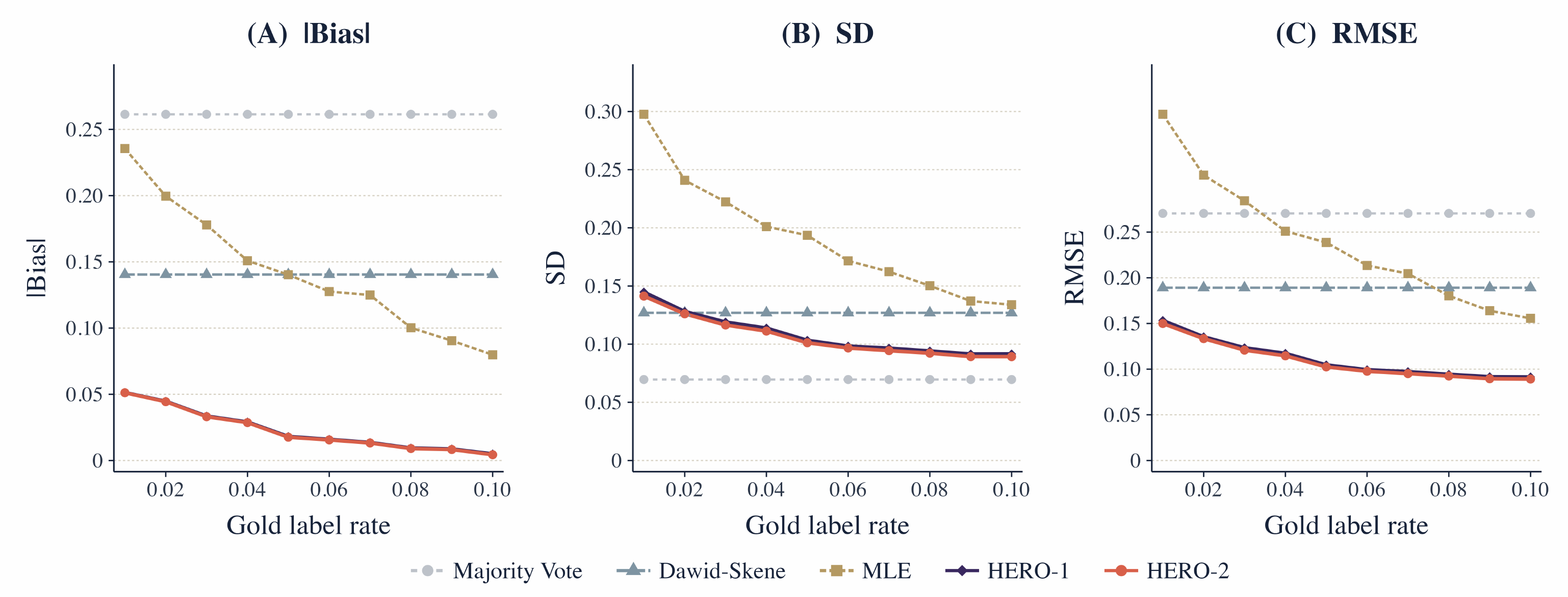}
    \caption{Experiment 1 Setting (ii): Comparison of (A) absolute bias, (B) variance, and (C) MSE with respect to various gold label rate. $\beta_1=1$, $J=5$, $n=150$, $N=400$
    {Effect of historical gold-label rate in Setting~(i). Panels report (A) absolute bias, (B) Monte Carlo standard deviation, and (C) RMSE for estimating $\mu_0=\mathbb{E}[Y_i]$. The current sample is small and the covariate signal is moderate: $\beta_1=1$, $J=5$, $n=150$, and $N=400$.}}
    \label{fig:sim-gold-small}
\end{figure}

\section{Additional details on real-world case studies}
\begin{enumerate}
    \item \textbf{Asset defect detection dataset.} 
    This dataset is designed to evaluate generation defects in 3D generation models. It contains data from two versions of the model: the historical dataset consists of outputs generated by the older model version, while the current dataset consists of outputs generated by the newer model version. Both model versions generate 3D assets from the same set of 532 prompts. For each prompt--asset pair, silver labelers evaluate whether the generated asset contains a specific type of defect. The silver labels are binary, with Yes/No responses encoded as 1/0. There are 10 silver labelers in total, and each prompt--asset pair is evaluated by three randomly selected, distinct silver labelers. In addition, a completely random 10\% subset of prompt--asset pairs is evaluated by a golden labeler, who is a well-trained expert. Most prompt--asset pairs are also evaluated by an AI labeler, GPT-5, which is prompted to provide the same binary defect label.

    \textit{Shared control variate}. For the HERO-2 method, we use the AI labels generated by GPT-5 as control variates. We repeat this procedure 100 times and report the absolute bias, standard deviation, and root mean squared error (RMSE) based on the resampled results, following the similar evaluation protocol as in the Simulation Study.

    \item \textbf{Safety evaluation dataset.} 
    This dataset contains both a historical dataset and a current dataset for evaluating the safety of a 3D generation model with respect to policy violation categories. The historical dataset corresponds to evaluations along one safety dimension, whereas the current dataset corresponds to evaluations along another safety dimension. The model generates 3D assets from 2,046 prompts. For each prompt--asset pair, silver labelers provide binary safety labels along the corresponding safety dimension. There are 11 silver labelers in total, and each prompt--asset pair is evaluated by a randomly selected set of 4--7 distinct silver labelers. In addition, a completely random 80\% subset of prompt--asset pairs is evaluated by a golden labeler, who is a well-trained expert, to benchmark the performance of the silver labelers.
    
    \textit{Shared control variate.} For HERO-2,  we use TF-IDF features extracted from the prompts as control variates. Since TF-IDF  (Term Frequency-Inverse Document Frequency) features are typically high-dimensional relative to the sample size, we apply a ridge penalty when estimating the coefficient $\gamma$ in the score-based control-variate adjustment. As in the previous case study, we report the absolute bias, standard deviation, and RMSE for HERO and the baseline methods.
\end{enumerate}

\section{Comparison with existing literature}
In this section, we compare HERO with three related strands of work: learning from noisy annotators, variance reduction techniques, and generative AI evaluation.

Learning from crowds and noisy annotators has a long history in statistics and machine learning. The classical Dawid–Skene (DS) model \cite{DawidSkene1979} introduced an EM-based approach to estimate individual annotator error rates from noisy labels. Subsequent work extended this idea to more complex scenarios, including latent variable models for multi-dimensional annotation \citep{welinder2010multidimensional}, probabilistic truth inference \citep{whitehill2009whose}, and crowdsourced aggregation with worker reliability \citep{raykar2010learning,sheng2008get}. Most of these studies assume that all labelers are noisy, and they do not exploit the existence of a small set of high-quality gold labels for calibration, leaving a gap for methods that systematically combine historical gold-labeled subset with current noisy annotations.

Variance reduction has been studied extensively in statistics and experimental design. Classical tools include control variates and regression adjustment, where auxiliary information is used to reduce estimation variance without bias \citep{owen2013monte,lin2013agnostic}. In applied fields, \cite{deng2013improving} uses pre‑experiment measurements to improve sensitivity in controlled experiments, and Monte Carlo methods with control variates have a rich theoretical literature \citep{glasserman2004monte}. These techniques reduce variance by exploiting correlated auxiliary variables, but, to our knowledge, they have not been systematically adapted to the noisy‑label setting.

Evaluation of generative AI raises additional challenges. Human preference judgments and pairwise comparison metrics have been used in practice to assess LLM quality \citep{zheng2023judging,chiang2024chatbot}, and broad benchmarking frameworks such as HELM \citep{liang2022holistic} aggregate diverse evaluation metrics. Safety and robustness evaluations, such as toxicity benchmarks, further illustrate the complexity of human annotations \citep{gehman2020realtoxicityprompts,borkan2019nuanced}. However, there is not yet a fully statistical framework that jointly accounts for annotator heterogeneity, leverages historical gold supervision, and reduces variance in large‑scale AI model evaluation.

\end{document}